\documentclass[letterpaper, 10 pt, conference]{ieeeconf}
\IEEEoverridecommandlockouts
\usepackage{amsmath,amssymb,amsfonts}
\usepackage{algorithmic}
\usepackage{algorithm,algorithmic}
\usepackage{textcomp}
\usepackage{bbold}
\usepackage{url}
\usepackage{cite}
\allowdisplaybreaks
\usepackage[hidelinks]{hyperref}
\usepackage[noabbrev,capitalise,nameinlink]{cleveref}
\usepackage{nicefrac, xfrac}

\usepackage{graphicx}
\graphicspath{ {./Figures/} }
\usepackage{subcaption}

\usepackage{siunitx}
\DeclareSIUnit\year{yr}
\NewDocumentCommand\SIpi{}{\text{\ensuremath{\pi}}}
\usepackage[dvipsnames]{xcolor}

\usepackage[draft]{changes}

\usepackage{aliascnt}
\newtheorem{theorem}{Theorem}

\newaliascnt{corollary}{theorem}

\aliascntresetthe{corollary}
\crefname{corollary}{Corollary}{Corollaries}

\newaliascnt{lemma}{theorem}
\newtheorem{lemma}[lemma]{Lemma}
\aliascntresetthe{lemma}
\crefname{lemma}{Lemma}{Lemmas}

\newtheorem{remark}{Remark}
\newtheorem{assumption}{Assumption}

\usepackage[short]{optidef}

\usepackage{tikz}
\usetikzlibrary{positioning}
\usetikzlibrary{matrix}
\usetikzlibrary {arrows.meta}
\usetikzlibrary{shapes}
\usetikzlibrary{calc}
\usetikzlibrary{math}
\usetikzlibrary{decorations.markings}
\usetikzlibrary{angles,quotes}
\usepackage[european,straightvoltages]{circuitikz}

\newcommand{\lha}[1]{\textcolor{black}{#1}} 
\newcommand{\lhb}[1]{\textcolor{black}{#1}} 
\newcommand{\lhc}[1]{\textcolor{black}{#1}} 
\newcommand{\icl}[1]{\textcolor{black}{#1}} 

\title{\LARGE \bf
Unstable Poles Arising 
\icl{in AC Power Grid} Subsystem Representations}

\author{Liam Hallinan and Ioannis Lestas%
    \thanks{L. Hallinan and I. Lestas are with the Department of Engineering, University of Cambridge, Trumpington Street, Cambridge, CB2 1PZ, United Kingdom. Emails:
        {\tt\small <lh706, icl20>@cam.ac.uk}.}%
}

\begin{document}
\maketitle

\begin{abstract}
    Recent small-signal stability studies of AC grids have shifted towards analysing power systems as interconnections of 
    \icl{subsystems} and leveraging their input-output properties to derive scalable stability certificates. Two subsystem representations appear frequently in the literature: the $PQ$ model, coupling powers to phase angle and voltage magnitude, and the $IV$ model, coupling \lhb{currents to voltages}. In this paper, we derive both models without simplifying the bus or line dynamics and show that a loop transformation relates the two. \icl{One of the main results in the paper is to} 
    then show analytically that each representation may \lhb{exhibit unstable poles} depending primarily on the operating point ($IV$ model) or the presence of high-frequency passive dynamics ($PQ$ model). \icl{In particular, such unstable poles in the subsystems can occur even when the aggregate interconnection is stable and well-behaved.}
    These effects are validated numerically, including \lhb{a case study} using the full-order dynamics of a synchronous generator with an exciter and transformer. Our results highlight that care must be taken when choosing a subsystem representation, as neglecting high-frequency dynamics or device operating points may obscure unstable poles that must be stabilised by the network interconnection and \lhb{must be} accounted for in system identification.
\end{abstract}

\section{Introduction}
Power systems are undergoing a profound transformation as synchronous generation is rapidly replaced by grid-forming power electronics, fundamentally changing how the grid is controlled and how its stability is analysed \cite{milano_FoundationsChallenges_18}.
The highly distributed nature of the evolving grid means that classical small-signal techniques, such as eigenvalue analysis applied to fixed, monolithic, system-wide models, are no longer appropriate.

To address these challenges, attention has recently shifted towards analysing power systems as interconnections of distributed subsystems. By studying the open-loop input-output properties of individual subsystems and their interconnection structure, recent work has aimed to certify small-signal stability by imposing scalable constraints on the allowable subsystem dynamics \cite{huang_GainPhase_24, chen_ExtendedFrequencyDomain_24, haberle_DecentralizedParametric_25, watson_ScalableControl_21}. These subsystems are often analysed in the frequency domain, increasing the importance of system identification for capturing complex, fast-timescale device dynamics.

However, the representation of these subsystems is non-unique in general and depends primarily on the port variables chosen to facilitate the interconnection. Two representations commonly arise in the literature using natural coupling variables in AC power systems. The first links active and reactive powers to electrical phase angle and voltage magnitude (the $PQ$ model) and typically appears in the power systems literature with simplified device dynamics. The second links current and voltage (the $IV$ model), is favoured within the power electronics community, and often incorporates more complex dynamics.

In this paper, we argue that the choice of coupling variables plays a significant role when analysing the network interconnection and that care should be taken when using either representation to analyse small-signal stability. In particular, we show that \icl{unstable poles}
may arise in the subsystem transfer functions under common device scenarios in either configuration. The contributions of our paper are as follows:
\begin{enumerate}
    \item We derive two equivalent small-signal representations of the grid as the interconnection of local subsystems: the $PQ$ and $IV$ models. Both are derived
    \lhb{without simplifying assumptions}
    on the bus or line dynamics.
    In particular, we show that a loop transformation can be used to translate from one representation to the other.
    \item We then show analytically, using practically relevant scenarios, that stable subsystems in the $PQ$ model can exhibit instabilities in the $IV$ model, and vice versa. This is revealed using the multivariable Nyquist criterion and is validated using numerical examples. Specifically, an unstable pole may appear in the $IV$ representation \lhb{of a droop-controlled bus} depending primarily on the device operating point, while unstable poles may appear in the $PQ$ model when the device \lhb{is placed in series with a passive impedance}.
    \item We further demonstrate that these effects appear when using a more complex synchronous generator model with exciters and transformers in a multi-machine case study.
\end{enumerate}

These results demonstrate that the choice of subsystem representation should be made with care; in particular, neglecting high-frequency dynamics or the device operating point when performing stability studies or system identification may obscure unstable poles that must be stabilised by the interconnection.

This paper is organised as follows. In \cref{sec:background}, we 
\icl{provide an overview of} necessary mathematical theory used to derive our results. In \cref{sec:model}, we present the complete and unsimplified small-signal $IV$ and $PQ$ grid models, and show how a loop transformation can be used to translate between \lhc{the two representations}. Then, in \cref{sec:instabilities}, we present the analysis and numerical results that demonstrate how instabilities may arise in each configuration under common device scenarios. Finally, we conclude in \cref{sec:conclusion}.

\section{Mathematical Background} \label{sec:background}
\subsection{Notation and Definitions}
Let $\mathbb{R}$ denote the reals, $\mathbb{C}$ denote the complex plane, \lha{$j\mathbb{R}$ denote the imaginary axis}, $\mathbb{C}_+ = \{s\in\mathbb{C}:\Re (s)> 0\}$ denote the open right half-plane, and $\bar{\mathbb{C}}_+ = \mathbb{C}_+ \cup j\mathbb{R} = \{s\in\mathbb{C}:\Re (s)\ge 0\}$ denote the closed right half-plane.

Let $I_n$ and $0_n$ denote the $n \times n$ identity and zero matrices, respectively (the dimension will be omitted when it is clear from the context).
Let $J = \begin{bsmallmatrix} 0 & 1 \\ -1 & 0\end{bsmallmatrix}$, which corresponds to clockwise rotation by $90^\circ$ in $\mathbb{R}^2$.
The Kronecker product of two matrices $A$ and $B$ is denoted $A \otimes B$. For a matrix $M \in \mathbb{C}^{n \times n}$, we denote its spectrum (set of eigenvalues) by $\sigma(M) = \{\lambda_1, \ldots,
\lambda_n\}$.
Let $\mathcal{K}=\{\kappa_1, \ldots, \kappa_{|\mathcal{K}|}\}$ be an ordered set of indexed elements. Then, the direct sum of a set of indexed matrices $B_k \in \mathbb{C}^{m_k \times n_k}$ associated with each $\kappa_k \in \mathcal{K}$ is denoted by $\oplus_{\kappa_k \in \mathcal{K}} B_k = \mathrm{diag} (B_1, \ldots, B_{|\mathcal{K}|})$. In addition, the composite vector constructed from a set of indexed vectors $a_k \in \mathbb{C}^{n_k}$ associated with each $\kappa_k \in \mathcal{K}$ is denoted $[a_k]_{\kappa_k \in \mathcal{K}} =[a_1^T, \ldots,a_{|\mathcal{K}|}^T]^T$.

\subsection{Nyquist Theory}

Let $P(s)$ and $K(s)$ be two proper \icl{real rational} transfer-function matrices \lhb{of dimension $m\times n$ and $n \times m$, respectively}, and let $[K(s),P(s)]$ denote their negative-feedback interconnection. 
\icl{We say} $[K(s),P(s)]$ is \textit{well-posed} if all closed-loop transfer matrices are well defined and proper \cite{zhou_RobustOptimal_96}. \icl{We also say the interconnection is stable if $(I+K(s)P(s))^{-1}$ has no poles in $\bar{\mathbb{C}}_+$. This also implies internal stability if there are no pole \lhc{cancellations} in $\bar{\mathbb{C}}_+$ in the product $K(s)P(s)$.}

The generalised Nyquist stability theorem \cite{maciejowski_MultivariableFeedback_89, desoer_GeneralizedNyquist_80} gives a necessary and sufficient condition for the feedback configuration $[K(s),P(s)]$ to be stable. We define the \textit{Nyquist contour}, illustrated in \cref{fig:tikz_nyquist_std}, as the closed curve on the complex plane containing the imaginary axis, and an infinite\footnote{\icl{The limit $R\to\infty$ in \eqref{eq:nyquist_contour} and \eqref{eq:nyquist_contour_mod} is considered in the compactified complex plane.}} semicircle in the right half-plane joining the positive imaginary axis to the negative imaginary axis \cite{skogestad_MultivariableFeedback_05}:
\begin{equation} \label{eq:nyquist_contour}
    \Gamma_N = \lim_{R \rightarrow \infty} \left\{ \Gamma_{j\omega} \cup \Gamma_R \right\},
\end{equation}
where
        $ \Gamma_{j\omega} = \left\{ s = j\omega \ | \ \omega \in \left[-R,R \right] \right\}$,
        $ \Gamma_R = \left\{s =  Re^{j \varphi} \ \big| \ \varphi \in \left[ \frac{\pi}{2},- \frac{\pi}{2} \right] \right\}$.
Furthermore, we define the \textit{modified Nyquist contour} with an indentation around the origin in $\mathbb{C}_+$ as
\begin{equation} \label{eq:nyquist_contour_mod}
        \Gamma_N^{\mathrm{mod}} = \lim_{\icl{R \rightarrow \infty}}
        \left\{ \Gamma_{j\omega}^- \cup \Gamma_\epsilon \cup \Gamma_{j\omega}^+ \cup \Gamma_R \right\},
    \end{equation}
where
    $ \Gamma_{j\omega}^- = \left\{ s = j\omega \ | \ \omega \in \left[-R,-\epsilon \right] \right\}$,
    $ \Gamma_\epsilon = \left\{s =  \epsilon e^{j \varphi} \ \big| \ \varphi \in \left[ -\frac{\pi}{2}, \frac{\pi}{2} \right] \right\}$,
    $ \Gamma_{j\omega}^+ = \left\{ s = j\omega \ | \ \omega \in \left[\epsilon,R \right] \right\}$.
\icl{and $\epsilon>0$ is sufficiently small.}

\begin{figure}[t]
    \centering
    \begin{subfigure}{0.2\textwidth}
        \centering
        \begin{tikzpicture}[>=Latex, line cap=round,scale=0.4]
\small
\def\R{3.2}      
\def\eps{0}   

\tikzset{
  nyqpath/.style={
    very thick,
    postaction={decorate},
    decoration={markings,
      mark=at position 0.15 with {\arrow{>}},
      mark=at position 0.28 with {\arrow{>}},
      mark=at position 0.55 with {\arrow{>}},
      mark=at position 0.9  with {\arrow{>}}
    }
  }
}

\draw[->] (-1.0,0) -- (\R+0.9,0) node[below right] {$\Re$};
\draw[->] (0,-\R-0.9) -- (0,\R+0.9) node[above left] {$\Im$};

\node[above left] at (0,\R) {$j\infty$};
\node[below left] at (0,-\R) {$-j\infty$};

\draw[nyqpath]
  (0,-\R)
    -- (0,\eps)
    arc[start angle=-90, end angle=90, radius=\eps cm]
    -- (0,\R)
    arc[start angle=90, end angle=-90, radius=\R cm];

\node at (\R*1.2,0.5) {$\Gamma_R$};
\node at (-0.75,\R*0.75) {$\Gamma_{j\omega}$};


\end{tikzpicture}
        \caption{$\Gamma_N$.}
        \label{fig:tikz_nyquist_std}
    \end{subfigure}
    \begin{subfigure}{0.2\textwidth}
        \centering
        \begin{tikzpicture}[>=Latex, line cap=round,scale=0.4]
\small
\def\R{3.2}      
\def\eps{0.5}   

\tikzset{
  nyqpath/.style={
    very thick,
    postaction={decorate},
    decoration={markings,
      mark=at position 0.05 with {\arrow{>}},
      mark=at position 0.35 with {\arrow{>}},
      mark=at position 0.55 with {\arrow{>}},
      mark=at position 0.9  with {\arrow{>}}
    }
  }
}

\draw[->] (-1.0,0) -- (\R+0.9,0) node[below right] {$\Re$};
\draw[->] (0,-\R-0.9) -- (0,\R+0.9) node[above left] {$\Im$};

\node[above left] at (0,\R) {$j\infty$};
\node[below left] at (0,-\R) {$-j\infty$};

\draw[nyqpath]
  (0,-\R)
    -- (0,-\eps)
    arc[start angle=-90, end angle=90, radius=\eps cm]
    -- (0,\R)
    arc[start angle=90, end angle=-90, radius=\R cm];

\node at (\R*1.2,0.5) {$\Gamma_R$};
\node at (2*\eps,0.5) {$\Gamma_{\epsilon}$};
\node at (-0.75,\R*0.75) {$\Gamma_{j\omega}^+$};
\node at (-0.75,-\R*0.75) {$\Gamma_{j\omega}^-$};


\end{tikzpicture}
        \caption{$\Gamma_N^{\mathrm{mod}}$.}
        \label{fig:tikz_nyquist_mod}
    \end{subfigure}
    \caption{Illustrations of the Nyquist contour and modified Nyquist contour with an indentation at the origin.}
    \label{fig:tikz_nyquist}
\end{figure}
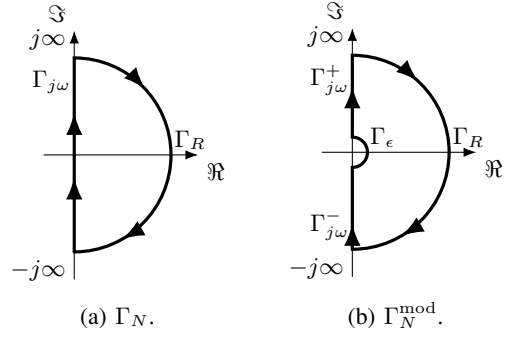

\begin{theorem} \label{thm:nyquist_std}
    Assume the closed-loop system $[K(s),P(s)]$ is well-posed. Let $L(s) = K(s)P(s)$ be the open-loop transfer function and let $L(s)$ have $P_{ol}$ poles in $\mathbb{C}_+$. If $L(s)$ has no poles on $j\mathbb{R}$, let $s$ traverse $\Gamma_N$ in a clockwise direction, and if $L(s)$ has no poles on $j\mathbb{R}$ other than a simple pole at $0$, let $s$ traverse $\Gamma_N^\mathrm{mod}$. Let $N_{\Gamma}$ denote the net number of clockwise encirclements made by the plots of $\sigma(L(s))$ around the point $-1$ as $s$ traverses either $\Gamma_N$ or $\Gamma_N^\mathrm{mod}$ \icl{with $\epsilon\to0$.}
    Then
    \begin{equation} \label{eq:nyquist_pole_count}
        N_\Gamma = \lhb{P_{cl} - P_{ol}},
    \end{equation}
    where $P_{cl}$ gives the number of \icl{poles         
    in $\mathbb{C}_+$ of the closed-loop transfer function \icl{$(I+K(s)P(s))^{-1}$}.}

\end{theorem}


\section{Grid Model} \label{sec:model}
In this section, we present two equivalent formulations of the \icl{AC grid} model used for small-signal stability analysis: \lha{the $IV$ and $PQ$ models}.
We first introduce \lhb{the} graph structure and reference frames used to \lhb{model the network}. 

\subsection{Topology and reference frames}

The three-phase AC grid is composed of $N_B$ buses, which form the node set $\mathcal{V}_B := \{\nu_1, \ldots, \nu_{N_B}\}$. In addition, we include a node representing ground, denoted $\nu_0$. The buses in $\mathcal{V}_B$ are connected by a set of $N_P$ power lines, which are assigned to the edge set $\mathcal{E}_P \subseteq \mathcal{V}_B \times \mathcal{V}_B$ with an arbitrary orientation. Denote the line from $\nu_j$ (source node) to $\nu_i$ (sink node) by $(\nu_i,\nu_j) \in \mathcal{E}_P$. Loads and other shunt connections link elements of $\mathcal{V}_B$ to the ground node $\nu_0$, and are assigned to the set $\mathcal{E}_0 \subseteq \mathcal{V}_B \times \{\nu_0\}$. If node $\nu_i\in\mathcal{V}_B$ has a load or shunt connection, the edge is directed so that $(\nu_i,\nu_0) \in \mathcal{E}_0$.

Together, the power grid is represented by the directed graph $(\mathcal{V}_N,\mathcal{E}_N)$, where $\mathcal{V}_N = \mathcal{V}_B \cup\{\nu_0\}$ and $\mathcal{E}_N = \mathcal{E}_P \cup \mathcal{E}_0 $.
\begin{assumption} \label{assump:connected}
    The subgraph $(\mathcal{V}_B,\mathcal{E}_P) \subseteq (\mathcal{V}_N,\mathcal{E}_N)$ is connected.
\end{assumption}

For each $\nu_i \in \mathcal{V}_N$, let $\mathcal{N}_i \subseteq \mathcal{V}_N$ be the set of neighbours of the node $\nu_i$. We also define $\mathcal{N}_i^+ \subseteq \mathcal{N}_i$ as the set of nodes $\nu_j$ such that $(\nu_i,\nu_j)\in\mathcal{E}_N$ (i.e., $\nu_i$ is the sink node), and $\mathcal{N}_i^- \subseteq \mathcal{N}_i$ as the set of nodes $\nu_j$ such that $(\nu_j,\nu_i) \in \mathcal{E}_N$ (i.e., $\nu_i$ is the source node).

For each edge $(\nu_i,\nu_j) \in \mathcal{E}_P$ (i.e., the power lines only), assign an index $k$ so that $\varepsilon_k \equiv (\nu_i,\nu_j)$ and order the set so that $\mathcal{E}_P = \{\varepsilon_1,\ldots,\varepsilon_{N_P}\}$. The interconnection structure of the power system is then captured by the incidence matrix $\mathcal{B}_N\in\mathbb{R}^{N_B \times N_P}$ for the bus--line subgraph $(\mathcal{V}_B,\mathcal{E}_P)$, with $(i,k)^{\mathrm{th}}$ entry given by
\begin{equation} \label{eq:incidence_matrix}
    \mathcal{B}_N^{(i,k)} = \begin{cases}
        1, & \text{if $\varepsilon_k \equiv (\nu_i,\nu_j) \in \mathcal{E}_P$, for some $\nu_j$}, \\
        -1, & \text{if $\varepsilon_k \equiv (\nu_j,\nu_i) \in \mathcal{E}_P$, for some $\nu_j$}, \\
        0, & \text{otherwise.}
    \end{cases}
\end{equation}

At each bus and each edge of the power system $(\mathcal{V}_N,\mathcal{E}_N)$, we associate a three-phase voltage $v_i(t) \in \mathbb{R}^3$ and net current injection $i_i(t) \in \mathbb{R}^3$ whose dynamics are defined by the bus and network system discussed in the following sections. We make the following assumption for the AC power system.
\begin{assumption}
    All currents and voltages are balanced three-phase signals.
\end{assumption}
Therefore all currents and voltages can be projected into a set of synchronously rotating $DQ$ reference frames
\cite{kundur_PowerSystem_22, schiffer_SurveyModeling_16, spanias_SystemReference_19}.

We first define a \textit{common} $DQ$ reference frame rotating at constant angular frequency $\omega_0 >0$, with instantaneous electrical phase angle $\theta_0(t) \in [0,2\pi)$ defining the angle between the $D$-axis and a fixed reference frame. A balanced three-phase signal $x(t) \in \mathbb{R}^3$ evaluated in this frame is denoted $x^{DQ}(t)=[x^D(t),x^Q(t)]^T \in \mathbb{R}^2$.
Similarly, with each bus $\nu_i \in \mathcal{V}_B$, we associate a \textit{local} $dq$ reference frame rotating at the dynamic local angular frequency $\omega_i(t)>0$, with instantaneous electrical phase angle $\theta_i(t) \in [0,2\pi)$ giving the angle between the $d$-axis and the fixed reference frame. The same signal $x(t)$ can be translated into the local $dq$ frame associated with $\nu_i\in\mathcal{V}_B$ to give $x^{dq}(t)=[x^d(t),x^q(t)]^T \in \mathbb{R}^2$.

The difference in phase angles between the local reference frame associated with $\nu_i \in \mathcal{V}_B$ and the common reference frame is
$ 
    \delta_i(t) = \theta_i(t) - \theta_0(t),
$ 
giving
\begin{equation} \label{eq:delta_dot}
    \dot{\delta}_i(t) = \omega_i(t) - \omega_0,
\end{equation}
\lhb{which implies} that $\omega_i(t) = \omega_0$ at equilibrium.
Therefore, we can map between the reference frames as
$ 
    x^{DQ}(t) = T(\delta_i(t))x^{dq}(t),
$ 
where
\begin{equation} \label{eq:rotation_matrix}
    T(\delta_i(t)) = \begin{bmatrix}
        \cos(\delta_i(t)) & -\sin(\delta_i(t)) \\ \sin(\delta_i(t)) & \cos(\delta_i(t))
    \end{bmatrix}.
\end{equation}
Without loss of generality, we \lhb{define} the local $d$-axis \lhb{to be aligned} with the bus voltage vector so that $v_i^{dq}(t) = [V_i(t), 0]^T$, where $V_i(t) > 0$ is the voltage magnitude. Therefore, using \eqref{eq:rotation_matrix}, we obtain
\begin{equation} \label{eq:vDQ_in_polar}
    \begin{bmatrix}
        v_i^D(t) \\ v_i^Q(t)
    \end{bmatrix} = \begin{bmatrix}
        V_i(t) \cos(\delta_i(t)) \\ V_i(t) \sin(\delta_i(t))
    \end{bmatrix}.
\end{equation}
Unless it is otherwise ambiguous, we henceforth omit the argument $t$ from time-domain signals.


Finally, we define the net active- and reactive-power flows at bus $\nu_i \in \mathcal{V}_B$, denoted $P_i$ and $Q_i$, respectively, as
\begin{subequations} \label{eq:powers}
    \begin{align}
        P_i &= {v_i^{DQ}}^T i_i^{DQ} = v_i^Di_i^D + v_i^Qi_i^Q, \label{eq:active_power} \\
        Q_i & = -{v_i^{DQ}}^T J i_i^{DQ} = v_i^Qi_i^D - v_i^Di_i^Q. \label{eq:reactive_power}
    \end{align}
\end{subequations}
We note that $P_i$ and $Q_i$ are invariant under reference frame transformations as $T(\delta_i)^T = T(\delta_i)^{-1}$ \lhb{and $T(\delta_i)^T J T(\delta_i) = J$}.

\subsection{IV Model}
In this section, we model the AC power system as the feedback interconnection of \lhb{bus} impedances and a sparse network admittance using voltages and currents at the interconnection ports.

\subsubsection{Buses}
At each bus $\nu_i \in \mathcal{V}_N$, we consider a dynamical system whose output is the bus voltage in the common reference frame, denoted $v_i^{DQ} \in \mathbb{R}^2$. The input to each bus is the negative net current injection in the common reference frame, denoted $-i_i^{DQ} \in \mathbb{R}^2$, which is given in terms of the neighbouring branches using Kirchhoff's current law
\begin{equation} \label{eq:kcl}
    i_i^{DQ} = \textstyle \sum_{\nu_j \in \mathcal{N}_i^+} i_{ij}^{DQ} - \sum_{\nu_j \in \mathcal{N}_i^-} i_{ji}^{DQ},
\end{equation}
where $i_{ij}^{DQ} \in \mathbb{R}^2$ is the current through the branch $(\nu_i,\nu_j) \in \mathcal{E}_N$ flowing into $\nu_i$ from $\nu_j$.
Note that a double index on $i_{ij}^{DQ}$ denotes a branch current associated with the line/load $(\nu_i,\nu_j) \in \mathcal{E}_N$, while a single index on $i_i^{DQ}$ denotes a net nodal current injection associated with $\nu_i \in \mathcal{V}_N$.

The relationship between $-i_i^{DQ}$ and $v_i^{DQ}$ is
\lhb{determined by the dynamics of the connected device, }
such as a synchronous generator or a grid-forming inverter. As we are interested in the small-signal stability of the grid $(\mathcal{V}_N,\mathcal{E}_N)$, we linearise each system about its equilibrium operating point, its equilibrium input $-i_i^{DQ\star}$, and its equilibrium output $v_i^{DQ\star}$, taking into account the translation from local $dq$ frame to common $DQ$ frame. We define the bus impedance $Z_i(s)$ as the $2 \times 2$ transfer function for the following Laplace-domain input--output relation:
\begin{equation} \label{eq:bus_iv_relationship}
    \Delta v_i^{DQ}(s) = - Z_i(s) \Delta i_i^{DQ}(s).
\end{equation}
Here, $\Delta v_i^{DQ}(s)$ and $\Delta i_i^{DQ}(s)$ are the Laplace-domain deviations of $v_i^{DQ}$ and $i_i^{DQ}$ about their operating points, respectively.

As a special case, we model the ground node $\nu_0$ as \lhb{a} constant voltage source with zero output, meaning $\Delta v_0^{DQ}(s) = [0, 0]^T$.

\subsubsection{Network}
We associate with each line or load $(\nu_i,\nu_j) \in \mathcal{E}_N$ a subsystem whose input is the potential difference across the element in the common $DQ$ reference frame, given by the difference in the voltages of the adjacent buses $v_i^{DQ} - v_j^{DQ}$, and whose output is the associated branch current, $i_{ij}^{DQ}$. In the Laplace domain, we obtain
\begin{equation} \label{eq:line_dynamics_s}
    \Delta i_{ij}^{DQ}(s) = Y_{ij}(s) \left( \Delta v_i^{DQ}(s) - \Delta v_j^{DQ}(s) \right),
\end{equation}
where $Y_{ij}(s)$ is the branch admittance for the line $(\nu_i,\nu_j)\in\mathcal{E}_N$. For a line/load represented by series $RL$ components, we have
\begin{equation} \label{eq:line_admittance}
    Y_{ij}(s) = \begin{bmatrix}
        R_{ij}+X_{ij}\frac{s}{\omega_0} & - X_{ij} \\ X_{ij} & R_{ij}+X_{ij}\frac{s}{\omega_0}
    \end{bmatrix}^{-1},
\end{equation}
where $R_{ij} > 0$ is the line/load resistance and $X_{ij} > 0$ is the line/load reactance, both measured in p.u.

For a \lhb{load or shunt connection in} $\mathcal{E}_0$, we have $j=0$ and hence $\Delta v_j^{DQ}(s) = [0,0]^T$. We therefore define
\begin{equation} \label{eq:shunt_admittance}
    Y_i(s) \equiv Y_{i0}(s),
\end{equation}
where we set $Y_{i0}(s)$ to $0_2$ if the corresponding branch is absent.

Now, let $\Delta v_B^{DQ}(s) = [\Delta v_i^{DQ}(s)]_{\nu_i \in \mathcal{V}_B}$ and $\Delta i_B^{DQ}(s) = [\Delta i_i^{DQ}(s)]_{\nu_i \in \mathcal{V}_B}$. Using \eqref{eq:line_dynamics_s} in Kirchhoff's current law \eqref{eq:kcl} at each node $\nu_i \in \mathcal{V}_B$ and using the fact that $Y_{ji}(s) \equiv Y_{ij}(s)$ \lhb{(by symmetry in the line model)} allows us to write
\begin{equation} \label{eq:kcl_matrix_form}
    \Delta i_B^{DQ}(s) = Y_N(s) \Delta v_B^{DQ}(s),
\end{equation}
where $Y_N(s)$ is the network admittance matrix. The matrix $Y_N(s)$ is composed of $2 \times 2$ blocks, where the $(i,j)^{\mathrm{th}}$ block is
\begin{equation} \label{eq:network_admittance_blocks}
    Y_N^{(i,j)}(s) = \begin{cases}
          Y_i(s) + \!\!\!\!\!\! \sum\limits_{\nu_k \in \mathcal{N}_i \setminus \nu_0} \!\!\!\!\! Y_{ik}(s), & i=j \\
        -Y_{ij}(s), & \text{$(\nu_i,\nu_j)$ or $(\nu_j,\nu_i)\in \mathcal{E}_P$} \\
        0_2, & \text{otherwise,}
    \end{cases}
\end{equation}
where $Y_i(s)$ is given by \eqref{eq:shunt_admittance}. Equivalently,
\begin{equation} \label{eq:network_admittance}
    Y_N(s) = \mathcal{B}_{N2} Y_P(s) \mathcal{B}_{N2}^T + Y_0(s),
\end{equation}
where $\mathcal{B}_{N2} = \mathcal{B}_N \otimes I_2$ (with $\mathcal{B}_N$ defined in \eqref{eq:incidence_matrix}), $Y_P(s) = \oplus_{(\nu_i,\nu_j) \in \mathcal{E}_P} Y_{ij}(s)$ (with block order following that used for $\mathcal{B}_N$ in \eqref{eq:incidence_matrix}), and $Y_0(s) = \oplus_{\nu_i \in \mathcal{V}_B} Y_i(s)$, with $Y_i(s)$ defined in \eqref{eq:shunt_admittance}.


Now \lhb{we assemble the bus impedances into a block-diagonal system.} Let
    $ Z_B(s) = \oplus_{\nu_i \in \mathcal{V}_B} Z_i(s) $.
Then \eqref{eq:bus_iv_relationship} can be written in matrix form as
\begin{equation} \label{eq:kvl_matrix_form}
    \Delta v_B^{DQ} = - Z_B(s) \Delta i_B^{DQ}.
\end{equation}
Together, \eqref{eq:kcl_matrix_form} and \eqref{eq:kvl_matrix_form} form the negative-feedback interconnection $[Y_N(s), Z_B(s)]$ between the bus dynamics and the network admittance matrix, as illustrated in \cref{fig:tikz_model_iv}.

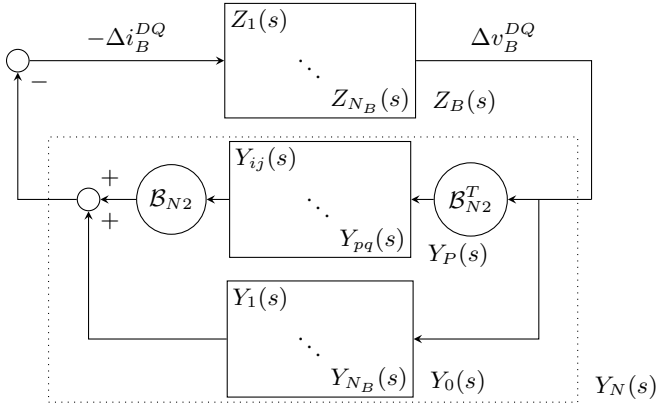
\begin{figure}[t]
\centering
\begin{tikzpicture}[scale=1]
\small

\node[matrix, draw,
    row sep=1pt, column sep=1pt,
    inner sep=1pt]
    (sys_Z) at (0,0)
    {
        \node{$Z_1(s)$}; & & \\
        & \node{$\ddots$}; & \\
        & & \node{$Z_{N_B}(s)$}; \\
    };
    \node at (sys_Z.south east)[right=3pt,yshift=6pt]{$Z_B(s)$};
\node[matrix, draw=black,
    below=of sys_Z,yshift=+20pt,
    row sep=1pt, column sep=1pt,
    inner sep=1pt]
    (sys_Yp) 
    {
        \node{$Y_{ij}(s)$}; & & \\
        & \node{$\ddots$}; & \\
        & & \node{$Y_{pq}(s)$}; \\
    };
    \node at (sys_Yp.south east)[right=3pt,yshift=1pt]{$Y_P(s)$};
\node[matrix, draw=black,
    below=of sys_Yp,yshift=+20pt,row sep=1pt, column sep=1pt,
    inner sep=1pt]
    (sys_Y0) 
    {
        \node{$Y_1(s)$}; & & \\
        & \node{$\ddots$}; & \\
        & & \node{$Y_{N_B}(s)$}; \\
    };
    \node at (sys_Y0.south east)[right=3pt,yshift=6pt]{$Y_0(s)$};

\node[draw, circle, minimum size=0.75cm,
            left=of sys_Yp,xshift=20pt]
            (B)  
            {$\mathcal{B}_{N2}$};
\node[draw, circle, minimum size=0.75cm,
            right=of sys_Yp,xshift=-20pt]
            (Bt)   {$\mathcal{B}_{N2}^T$};

\node[draw, circle, minimum size=0.05cm,
    left=of B,xshift=15pt]
    (left1) 
    {};
    \node[yshift=5pt,xshift=5pt] at (left1.north east){$+$};
    \node[yshift=-5pt,xshift=5pt] at (left1.south east){$+$};
\node[right= of Bt,xshift=-20pt] (right1){};
\node[right= of right1,xshift=-15pt] (right2){};
\node[draw, circle, minimum size=0.05cm]
    (neg) at (-4,0)
    {};
    \node[yshift=-5pt,xshift=5pt] at (neg.south east){$-$};

\node[draw,black,thin,dotted, anchor=north west,
    minimum width =7cm, minimum height =3.5cm]
    (box) at (-3.6,-1.0) {};
    \node at (box.south east)[right=3pt,yshift=6pt]{$Y_N(s)$};

\draw[-stealth] (neg.east) -- (sys_Z.west)
    node[midway,above]{$-\Delta i_B^{DQ}$};
\draw[] (sys_Z.east) -| (right2.center)
    node[pos=0.25,above]{$\Delta v_B^{DQ}$};
\draw[-stealth] (right2.center) -- (Bt.east);
\draw[-stealth] (right1.center) |- (sys_Y0.east);
\draw[-stealth] (Bt.west) -- (sys_Yp.east);
\draw[-stealth] (sys_Yp.west) -- (B.east);
\draw[-stealth] (B.west) -- (left1.east);
\draw[-stealth] (sys_Y0.west) -| (left1.south);
\draw[-stealth] (left1.west) -| (neg.south);

\end{tikzpicture}
\caption{Negative-feedback interconnection of the bus impedances in $Z_B(s)$ and the network admittance $Y_N(s)$.}
\label{fig:tikz_model_iv}
\end{figure}

\subsection{PQ Model}
As an alternative to the $IV$ model, we now model the grid as a feedback interconnection coupling the bus active and reactive powers to electrical phase angle and voltage magnitude.

First, at each $\nu_i \in \mathcal{V}_B$, we linearise \eqref{eq:powers} with respect to $v_i^{DQ}$ and $i_i^{DQ}$ and convert to the Laplace domain to get
\begin{equation} \label{eq:pq_deriv_1}
    \begin{aligned}
        \underbrace{\begin{bmatrix}
            \Delta P_i(s) \\ \Delta Q_i(s)
        \end{bmatrix}}_{\Delta S_i(s)}
        & = \begin{bmatrix}
            i_i^{D\star} & \!\!\!\!\! i_i^{Q\star}  \\ -i_i^{Q\star} & \!\!\!\!\! i_i^{D\star}
        \end{bmatrix} \Delta v_i^{DQ}(s) \! + \! \underbrace{\begin{bmatrix}
            v_i^{D\star} & \!\!\!\!\! v_i^{Q\star} \\ v_i^{Q\star} & \!\!\!\!\!- v_i^{D\star}
        \end{bmatrix}}_{U_i^\dagger} \Delta i_i^{DQ}(s),
    \end{aligned}
\end{equation}
where $(\cdot)^\star$ represents the equilibrium value of the respective variable. Next, linearising \eqref{eq:vDQ_in_polar} gives
\begin{equation} \label{eq:pq_deriv_2}
    \begin{aligned}
        \Delta v_i^{DQ}(s)
        & = \begin{bmatrix}
            -V_i^\star \sin(\delta_i^\star) & \cos(\delta_i^\star) \\
            V_i^\star \cos(\delta_i^\star) & \sin(\delta_i^\star)
        \end{bmatrix} \begin{bmatrix}
            \Delta \delta_i(s) \\ \Delta V_i(s)
        \end{bmatrix} \\
        & = \underbrace{\begin{bmatrix}
            -v_i^{Q\star} & v_i^{D\star} \\ v_i^{D\star} & v_i^{Q\star}
        \end{bmatrix}}_{U_i^\ddagger} \underbrace{\begin{bmatrix}
            \Delta \delta_i(s) \\ \frac{\Delta V_i(s)}{V_i^\star}
        \end{bmatrix}}_{\Delta \phi_i(s)}.
    \end{aligned}
\end{equation}
Using \eqref{eq:pq_deriv_2}, the first term in \eqref{eq:pq_deriv_1} gives
\begin{equation} \label{eq:pq_deriv_3}
    \begin{aligned}
        \begin{bmatrix}
                i_i^{D\star} & i_i^{Q\star} \\ -i_i^{Q\star} &  i_i^{D\star}
            \end{bmatrix} \Delta v_i^{DQ}(s)
            & = \underbrace{\begin{bmatrix}
                -Q_i^\star & P_i^\star \\ P_i^\star & Q_i^\star
            \end{bmatrix}}_{W_i} \Delta \phi_i(s).
    \end{aligned}
\end{equation}
Now, let $U^\dagger = \oplus_{\nu_i \in \mathcal{V}_B}U_i^\dagger$, $U^\ddagger = \oplus_{\nu_i \in \mathcal{V}_B}U_i^\ddagger$, $W = \oplus_{\nu_i \in \mathcal{V}_B}W_i$, $\Delta S(s) = [\Delta S_i(s)]_{\nu_i \in \mathcal{V}_B}$, and $\Delta \phi(s) = [\Delta \phi_i(s)]_{\nu_i \in \mathcal{V}_B}$. Then, using \eqref{eq:network_admittance}, the network dynamics defining the active and reactive power injections at each bus \lhb{are} given by
\begin{equation}
    \begin{split}
        \Delta S(s)
        & = U^\dagger Y_N(s) \Delta v^{DQ}(s) +  W \Delta \phi(s) \\
        & = (U^\dagger Y_N(s) U^\ddagger + W) \Delta \phi(s).
    \end{split}
\end{equation}
The network dynamics are therefore given by the transfer function
\begin{equation} \label{eq:N_def}
    N_{PQ}(s) := U^\dagger Y_N(s) U^\ddagger + W.
\end{equation}

\begin{figure}[t]
\centering
\begin{tikzpicture}
\small


\node[matrix, draw,
    row sep=1pt, column sep=1pt,
    inner sep=1pt]
    (sys_Z) at (0,0)
    {
        \node{$G_1(s)$}; & & \\
        & \node{$\ddots$}; & \\
        & & \node{$G_{N_B}(s)$}; \\
    };
    \node at (sys_Z.south east)[right=3pt,yshift=6pt]{$G_B(s)$};

\node[draw=black,
    below=of sys_Z,yshift=0.8cm,
    minimum size = 1cm]
    (sys_Yp) 
    {$Y_N(s)$};
    
\node[draw=black,
    below=of sys_Yp,yshift=0.5cm]
    (sys_Y0) 
    {$ \bigoplus\limits_{\nu_i \in \mathcal{V}_B} \begin{bsmallmatrix}
                -Q_i^\star & P_i^\star \\ P_i^\star & Q_i^\star
            \end{bsmallmatrix}$};
    \node at (sys_Y0.south east)[right=3pt,yshift=6pt]{$W$};

\node[draw=black,
    left=of sys_Yp,xshift=0.8cm]
    (B) 
    {$\bigoplus\limits_{\nu_i \in \mathcal{V}_B}\begin{bsmallmatrix}
                v_i^{D\star} & v_i^{Q\star} \\ v_i^{Q\star} & - v_i^{D\star}
            \end{bsmallmatrix}$};
    \node at (B.south east)[left=1pt,yshift=-6pt]{$U^\dagger$};
\node[draw=black,
    right=of sys_Yp,xshift=-0.8cm]
    (Bt) 
    {$\bigoplus\limits_{\nu_i \in \mathcal{V}_B} \begin{bsmallmatrix}
                -v_1^{Q\star} & v_1^{D\star} \\ v_1^{D\star} & v_1^{Q\star}
            \end{bsmallmatrix}$};
    \node at (Bt.south east)[left=1pt,yshift=-6pt]{$U^\ddagger$};

\node[draw, circle, minimum size=0.05cm,
    left=of B,xshift=0.5cm]
    (left1) 
    {};
    \node[yshift=5pt,xshift=5pt] at (left1.north east){$+$};
    \node[yshift=-5pt,xshift=5pt] at (left1.south east){$+$};
\node[right= of Bt,xshift=-0.8cm] (right1){};
\node[right= of right1,xshift=-1cm] (right2){};
\node[draw, circle, minimum size=0.05cm]
    (neg) at (-4.4,0)
    {};
    \node[yshift=-5pt,xshift=5pt] at (neg.south east){$-$};

\node[draw,black,thin,dotted,anchor=north west,
    minimum width =8.2cm, minimum height =2.6cm]
    (box) at (-4.3,-0.9) {};
    \node at (box.south east)[left=3pt,yshift=7pt]{$N_{PQ}(s)$};

\draw[-stealth] (neg.east) -- (sys_Z.west)
    node[pos=0.5,above]{$-\begin{bmatrix} \Delta P_i(s) \\ \Delta Q_i(s)\end{bmatrix}_{\nu_i \in \mathcal{V}_B}$};
\draw[] (sys_Z.east) -| (right2.center)
    node[pos=0.25,above]{$\begin{bmatrix} \Delta \delta_i(s) \\ \nicefrac{\Delta V_i(s)}{V_i^\star} \end{bmatrix}_{\nu_i \in \mathcal{V}_B}$};
\draw[-stealth] (right2.center) -- (Bt.east);
\draw[-stealth] (right1.center) |- (sys_Y0.east);
\draw[-stealth] (Bt.west) -- (sys_Yp.east);
\draw[-stealth] (sys_Yp.west) -- (B.east);
\draw[-stealth] (B.west) -- (left1.east);
\draw[-stealth] (sys_Y0.west) -| (left1.south);
\draw[-stealth] (left1.west) -| (neg.south);

\end{tikzpicture}
\caption{Negative-feedback interconnection of the bus systems in $G_B(s)$ and the network system $N_{PQ}(s)$.}
\label{fig:tikz_model_pq}
\end{figure}
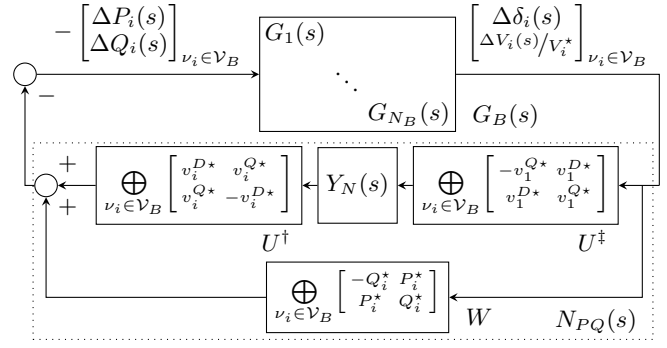

At each bus $\nu_i \in \mathcal{V}_B$, we let
\begin{equation} \label{eq:bus_pq_relationship}
    \begin{bmatrix}
        \Delta \delta_i(s) \\ \frac{\Delta V_i(s)}{V_i^\star}
    \end{bmatrix} = -G_i(s) \begin{bmatrix}
        \Delta P_i(s) \\ \Delta Q_i(s)
    \end{bmatrix},
\end{equation}
and let $G_B(s) = \oplus_{\nu_i \in \mathcal{V}_B} G_i(s)$. From \eqref{eq:delta_dot}, we have $\Delta \delta_i(s) =
\frac{1}{s} \Delta \omega_i(s)$, so $G_i(s)$ takes the form
\begin{equation} \label{eq:pq_transfer_function}
    G_i(s) = \begin{bmatrix}
        \frac{1}{s} G_{\omega p, i}(s) & \frac{1}{s} G_{\omega q,i}(s) \\ G_{vp,i}(s) & G_{vq,i}(s)
    \end{bmatrix}.
\end{equation}
The bus and network systems are then arranged to obtain the negative-feedback interconnection $[N_{PQ}(s),G_B(s)]$, as illustrated in \cref{fig:tikz_model_pq}.

Furthermore, we can link the $IV$ and $PQ$ transfer functions at each $\nu_i \in \mathcal{V}_B$ by substituting \eqref{eq:pq_deriv_1}-\eqref{eq:pq_deriv_3} into \eqref{eq:bus_pq_relationship} to obtain
$
    \Delta v_i^{DQ}(s) = - U_i^\ddagger (I+G_i(s)W_i)^{-1} G_i(s) U_i^\dagger \Delta i_i^{DQ}(s)
$,
which gives
\begin{equation} \label{eq:iv_from_pq}
    Z_i(s) = U_i^\ddagger (I+G_i(s)W_i)^{-1} G_i(s) U_i^\dagger.
\end{equation}
This shows that the $IV$ bus impedance is given by the negative feedback of the $PQ$ transfer function with the static matrix $W_i$ and scaled by the static matrices $U_i^\dagger$ and $U_i^\ddagger$, as illustrated by \cref{fig:tikz_pq_to_iv_fb}. Similarly, using $U^\dagger U^\dagger = U^\ddagger U^\ddagger = {V_i^\star}^2I_2$, we \lhb{invert \eqref{eq:iv_from_pq} to} obtain
\begin{equation} \label{eq:pq_from_iv}
    G_i(s) =  (I-\tilde{Z}_i(s) W_i)^{-1} \tilde{Z}_i(s),
\end{equation}
where $\tilde{Z}_i(s) = \frac{1}{{V_i^\star}^4}U_i^\ddagger Z_i(s)U_i^\dagger$, giving $G_i(s)$ as the positive-feedback interconnection of $\tilde{Z}_i(s)$ and $W_i$.
\lhb{We can therefore conclude that a translation between $PQ$ and $IV$ models can be performed by using a loop transformation involving the matrices $U_i^\dagger$, $U_i^\ddagger$, and $W_i$.}

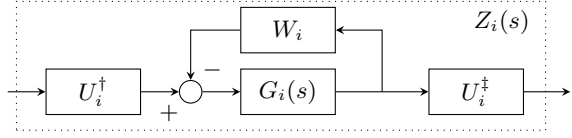
\begin{figure}[t]
\centering
\begin{tikzpicture}
\small

\node[draw,minimum width =1.25cm, minimum height =0.6cm]
    (sys_Z) at (0,0)
    {$G_i(s)$};

\node[draw=black,
    above=of sys_Z,yshift=-0.8cm,
    minimum width =1.25cm, minimum height =0.6cm]
    (sys_Y0) 
    {$W_i$};

\node[draw, circle, minimum size=0.05cm,
    left=of sys_Z,xshift=0.5cm]
    (left1) 
    {};
    \node[yshift=5pt,xshift=5pt] at (left1.north east){$-$};
    \node[yshift=-5pt,xshift=-5pt] at (left1.south west){$+$};
\node[draw=black,
    left=of left1,xshift=0.5cm,
    minimum width =1.25cm, minimum height =0.6cm]
    (B) 
    {$U_i^\dagger$};
\node[right= of sys_Z,xshift=-0.5cm] (right1){};
\node[draw=black,
    right=of right1,xshift=-0.5cm,
    minimum width =1.25cm, minimum height =0.6cm]
    (Bt) 
    {$U_i^\ddagger$};
\node[right= of Bt,xshift=-0.5cm] (right2){};
\node[left=of B,xshift=0.5cm]
    (neg)
    {};

\node[draw,black,thin,dotted,anchor=north west,
    minimum width =7cm, minimum height =1.7 cm]
    (box) at (-3.6,1.2) {};
    \node at (box.north east)[left=3pt,yshift=-8pt]{$Z_i(s)$};

\draw[-stealth] (neg.east) -- (B.west);
\draw[-stealth] (B.east) -- (left1.west);
\draw[-stealth] (left1.east) -- (sys_Z.west);
    
\draw[] (sys_Z.east) -- (right1.center);
\draw[-stealth] (right1.center) -- (Bt.west);
\draw[-stealth] (Bt.east) -- (right2.center);
\draw[-stealth] (right1.center) |- (sys_Y0.east);
\draw[-stealth] (sys_Y0.west) -| (left1.north);


\end{tikzpicture}
\caption{Relationship between the $PQ$ and $IV$ bus transfer functions.}
\label{fig:tikz_pq_to_iv_fb}
\end{figure}

\section{Unstable poles in \textit{IV} and \textit{PQ} models} \label{sec:instabilities}
We now show that both bus transfer functions in the $IV$ and $PQ$ models can exhibit unstable poles in common bus configurations. This adds complexity \lhb{when performing stability studies or system identification. In particular, Nyquist-based approaches to stability analysis}
require the Nyquist plot of the return ratios $L_{IV}(s) = Z_B(s)Y_N(s)$ or $L_{PQ}(s) = G_B(s)N_{PQ}(s)$ to make a corresponding number of encirclements of the point $-1$ to certify stability. \icl{The latter complicates the derivation of decentralised stability criteria which often relies on the subsystems being stable. The stabilisation of unstable poles in subsystems is achieved in many cases via interconnections with other subsystems thus bringing more centralised aspects in the analysis.}

\subsection{Low frequency poles in \textit{IV} impedances} \label{sec:mod_iv_instabilities}
Consider a bus $\nu_i \in \mathcal{V}_B$ that links frequency to active power and voltage magnitude to reactive power using a droop control law. The entries of the transfer function \eqref{eq:pq_transfer_function} for such a system are
\begin{equation} \label{eq:droop_tf}
    \begin{aligned}
        & G_{\omega p,i}(s) = \tfrac{k_{\omega,i}}{\tau_{\omega,i}s + 1}, &&
        G_{\omega q,i }(s) = 0, \\
        & G_{v p,i}(s) = 0, &&
        G_{v q,i}(s) = \tfrac{k_{v,i}}{\tau_{v,i}s + 1},
    \end{aligned}
\end{equation}
where $k_{\omega , i}, k_{v,i} >0$ are gains and $\tau_{\omega ,i}, \tau_{v,i} \geq 0$ are time constants.

As $G_i(s)$ in this case is marginally stable (due to the integrator in \eqref{eq:pq_transfer_function}), \lhb{using \cref{thm:nyquist_std}}, we can determine the number of right half-plane poles in $Z_i(s)$ \lhb{by counting} the number of encirclements of the point $-1$ by the characteristic loci of $L_{Z,i}(s) = G_i(s)W_i$ as $s$ traverses the modified Nyquist contour \lhb{$\Gamma_N^\mathrm{mod}$ in \eqref{eq:nyquist_contour_mod}.}
In the following Lemma, we evaluate the characteristic loci of $L_{Z,i}(s)$ as $s \rightarrow 0$.

\begin{lemma} \label{lemma:iv_unstable}
Consider a bus $\nu_i \in \mathcal{V}_B$ with $PQ$ transfer function $G_i(s)$ as in \eqref{eq:pq_transfer_function} with entries given by \eqref{eq:droop_tf}. As $s \rightarrow 0$, the eigenvalues of the return ratio $L_{Z,i}(s) = G_i(s)W_i$ with $W_i$ as in \eqref{eq:pq_deriv_3} are given by
\begin{equation} \label{eq:instab_iv_loci}
    \begin{aligned}
        \lambda_-(s) &= -\tfrac{Q_i^\star k_{\omega,i}}{s} + Q_i^\star k_{\omega,i}\tau_{\omega,i} - \tfrac{{P_i^\star}^2 k_{v,i}}{Q_i^\star} + \mathcal{O}(s), \\
        \lambda_+(s) &= Q_i^\star k_{v,i} + \tfrac{{P_i^\star}^2 k_{v,i}}{Q_i^\star} + \mathcal{O}(s).
    \end{aligned}
\end{equation}
\end{lemma}
\begin{proof}
In what follows, we omit the subscript $(\cdot)_i$ and superscript $(\cdot)^\star$ in the system parameters to simplify the notation. As $s\rightarrow 0$, the Taylor expansion of $G_i(s)$ gives
\begin{equation*}
    G_i(s) = \begin{bmatrix}
        \frac{k_\omega}{s} - k_\omega \tau_\omega + k_\omega \tau_\omega^2 s& 0 \\ 0 & k_v-k_v\tau_v s
    \end{bmatrix} + \mathcal{O}(s^2).
\end{equation*}
The eigenvalues of $L_{Z,i}(s)$ are given by
\begin{equation} \label{eq:lemma_eig_formula}
    \begin{aligned}
        \lambda_\pm(s) = & \tfrac{1}{2} \left(\mathrm{tr}(L_{Z}(s)) \right. \\
        & \qquad \left. \pm \sqrt{\mathrm{tr}(L_{Z}(s))^2 - 4\det(L_{Z}(s))}\right).
    \end{aligned}
\end{equation}
We have
\begin{equation*}
    \begin{aligned}
        \mathrm{tr}(L_{Z,i}(s)) =
            &-Qk_\omega\frac{1}{s} + Q(k_\omega \tau_\omega + k_v) \\
            &\qquad- Q(k_\omega \tau_\omega^2+ k_v\tau_v)s  + \mathcal{O}(s^2), \\
        \mathrm{tr}(L_{Z,i}(s))^2 =
            & Q^2k_\omega^2\frac{1}{s^2} - 2Q^{\lhc{2}}k_\omega(k_\omega \tau_\omega + k_v)\frac{1}{s} \\
            &\qquad + 2Q^{\lhc{2}}k_\omega(k_\omega\tau_\omega^2 + k_v\tau_v) \\
            &\qquad + \lhc{Q^2}(k_\omega \tau_\omega + k_v)^2 + \mathcal{O}(s), \\
                                \det(L_{Z,i}(s)) =
            & - \left(Q^2 + P^2\right)\left( k_\omega k_v \frac{1}{s} \right. \\
            & \qquad \left.- k_\omega k_v(\tau_\omega + \tau_v) \right) + \mathcal{O}(s).
    \end{aligned}
\end{equation*}
Therefore,
\begin{equation*}
    \begin{aligned}
        (\mathrm{tr}&(L_{Z,i}(s))^2  - 4\det(L_{Z,i}(s)) \\
        & = Q^2 k_\omega^2 \frac{1}{s^2} - 2 Q^2 k_\omega(k_\omega\tau_\omega - k_v)\frac{1}{s} + 4P^2 k_\omega k_v \frac{1}{s} \\
            &\qquad  + Q^2(k_\omega \tau_\omega - k_v)^2 +2Q^2 k_\omega(k_\omega\tau_\omega - k_v\tau_v) \\
            & \qquad -4P^2k_\omega k_v(\tau_\omega+\tau_v) + \mathcal{O}(s) \\
        & = Q^2\left(k_\omega \frac{1}{s} - (k_\omega \tau_\omega - k_v) \right)^2 + 4P^2k_\omega k_v \frac{1}{s} + \mathcal{O}(1).
    \end{aligned}
\end{equation*}
Let
$
    h(s) = Q\left(k_\omega \frac{1}{s} - (k_\omega \tau_\omega - k_v) \right)
$.
Then
\begin{equation*}
    \begin{aligned}
        (\mathrm{tr}(L_{Z,i}&(s))^2  - 4\det(L_{Z,i}(s)) \\
        & = h(s)^2 + 4P^2k_\omega k_v \frac{1}{s} + \mathcal{O}(1) \\
        & = h(s)^2 \left[1 + h(s)^{-2} \left(4P^2k_\omega k_v \frac{1}{s} + \mathcal{O}(s) \right) \right].
    \end{aligned}
\end{equation*}
Now
$
    h(s)^{-2} = \frac{s^2}{Q^2 k_\omega^2} + \mathcal{O}(s^3)
$,
so
\begin{equation*}
    \begin{aligned}
        (\mathrm{tr}(L_{Z,i}(s))^2  &- 4\det(L_{Z,i}(s)) \\
        & = h(s)^2 \left[1 + \frac{4P^2 k_v}{Q^2 k_\omega}s + \mathcal{O}(s^2) \right].
    \end{aligned}
\end{equation*}
Therefore
\begin{equation*}
    \begin{aligned}
        & \sqrt{ \mathrm{tr}(L_{Z,i}(s))^2  - 4\det(L_{Z,i}(s)) } \\
        & \qquad = h(s) \left[1 + \frac{2P^2 k_v}{Q^2 k_\omega}s + \mathcal{O}(s^2) \right] \\
        & \qquad = Qk_\omega \frac{1}{s} - Q k_\omega\tau_\omega + Q k_v + \frac{2P^2 k_v}{Q} + \mathcal{O}(s).
    \end{aligned}
\end{equation*}
Now, evaluating $\lambda_\pm(s)$ using \eqref{eq:lemma_eig_formula} gives \eqref{eq:instab_iv_loci}.
\end{proof}

We therefore see that the characteristic loci split into a bounded and unbounded branch as $s \rightarrow 0$. In particular, for $s = j\omega$, the unbounded branch approaches infinity along the imaginary axis (ignoring $\mathcal{O}(1)$ corrections), with direction determined by the sign of the equilibrium bus reactive power $Q_i^\star$ calculated via \eqref{eq:reactive_power}.

Along \lhb{$\Gamma_\epsilon$} (the infinitesimal indentation of the modified Nyquist contour around the origin),
$
    s = \epsilon e^{j\varphi},  -\frac{\pi}{2} \leq \varphi \leq \frac{\pi}{2},
$
so we obtain
\begin{equation}
    \lambda_-(\epsilon e^{j\varphi}) = -Q_i^\star k_{\omega,i} \epsilon^{-1} e^{-j\varphi} + \mathcal{O}(1).
\end{equation}
Clearly, if $Q_i^\star > 0$, then $\lambda_-(s)$ travels in a semi-circular arc in the left half-plane along the infinitesimal indentation around the origin. For reasonable values of $Q_i^\star$ and $k_{\omega,i}$, this arc is likely to result in a clockwise encirclement \lhb{of} the point $-1$, which (if not counteracted by a counter-clockwise encirclement at higher frequencies) implies that $Z_i(s)$ as calculated via \eqref{eq:iv_from_pq} has an unstable pole.

As an example, consider the case $k_{\omega,i} = 2$, $\tau_{\omega,i} = 1$, $k_{v,i} = 0.5$, $\tau_{v,i} = 2$, $P_i^\star = 1$ p.u., $Q_i^\star = 0.5$ p.u., and ${v_i^{DQ}}^\star = [1,0]^T$ p.u. Clearly $G_i(s)$ is marginally stable. The Nyquist plot of $L_{Z,i}(s) = G_i(s)W_i$ is shown in \cref{fig:plot_iv_unstable} and demonstrates that the unbounded branch encircles the point $-1$ as it traverses the modified Nyquist contour, indicating that the $IV$ model bus impedance $Z_i(s)$ is unstable. Through direct calculation, we find that $Z_i(s)$ has a right half-plane pole at $s = 0.7684$.

\begin{remark}
    Typically, the dynamics of a bus will be more complex than the simple model considered in \eqref{eq:droop_tf} and the full-order system should therefore be investigated to determine if there are any right half-plane poles when converting between the $PQ$ and $IV$ models. However, many bus systems exhibit a time-scale separation in their dynamics which means that the transfer functions \eqref{eq:droop_tf} are sufficiently accurate to describe the behaviour of the system at low frequencies. As the analysis of \cref{lemma:iv_unstable} involves the limit $s \rightarrow 0$, we can conclude that the dominant term in the expression for $\lambda_-(s)$ should hold for systems that are well approximated by \eqref{eq:droop_tf} at low frequencies, meaning the corresponding Nyquist plot may encircle the point $-1$ depending on the values of $Q_i^\ast$ and the other parameters of the system.
\end{remark}

\begin{figure}[t]
    \centering
    \includegraphics[width=0.48\textwidth]{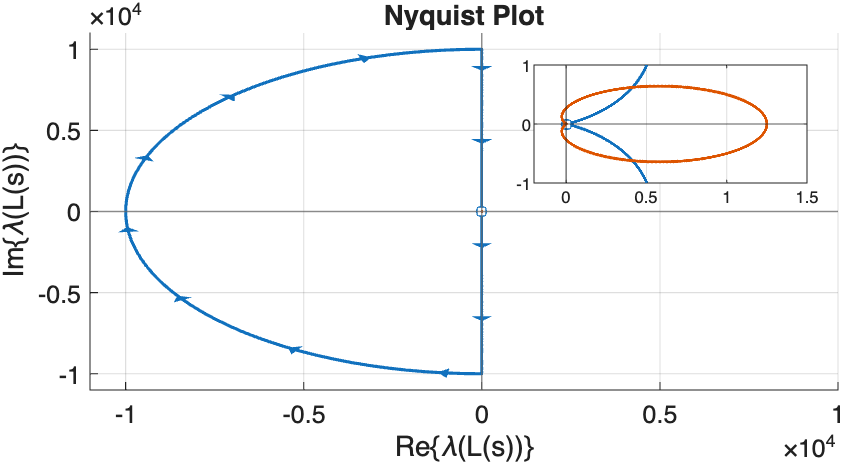}
    \caption{The Nyquist plot of $L_{Z,i}(s) = G_i(s)W_i$ along $\Gamma_N^\mathrm{mod}$ (with indentation of radius $\epsilon = 1\times 10^{-4}$) with entries of $G_i(s)$ given in \eqref{eq:droop_tf}. The inset rescales the plot to show the bounded branch.}
    \label{fig:plot_iv_unstable}
\end{figure}

\subsection{High frequency poles in \textit{PQ} systems}
Now consider a voltage source in series with a passive impedance consisting of $RLC$ components, e.g., a transformer or the low-pass filter of an \lhb{inverter}, as shown in \cref{fig:tikz_series_impedance}. If the voltage source has impedance $Z_i^s(s)$ and the passive impedance has transfer function $Z_i^p(s)$, then the total bus impedance is given by
\begin{equation}
    Z_i(s) = Z_i^s(s) + Z_i^p(s).
\end{equation}
Now, assume that over a certain frequency band (e.g., the resonant frequency of the filter) the passive impedance has much higher gain than the voltage source, so that
\begin{equation} \label{eq:Z_bus_approx_Z_passive}
    Z_i(j\omega) \approx Z_i^p(j\omega)
\end{equation}
in that frequency band. Such a passive impedance composed of $RLC$ elements takes the general form
\begin{equation} \label{eq:gen_passive_compent}
    Z_i^p(s)=\begin{bmatrix}
        a(s) & b(s) \\ -b(s) & a(s)
    \end{bmatrix};
\end{equation}
see, for example, the series $RL$ system given by $Y_{ij}(s)^{-1}$ in \eqref{eq:line_admittance}. As the system is passive, we know that $Z_i^p(s)$ has no unstable poles and is positive-real, i.e., $Z_i^p(j\omega) + Z_i^p(j\omega)^\ast \geq 0$, so that the eigenvalues of $Z_i^p(j\omega)$ have non-negative real part. We now show that when converted to a $PQ$ model using \eqref{eq:pq_from_iv}, the passive component \eqref{eq:gen_passive_compent} may have unstable poles.

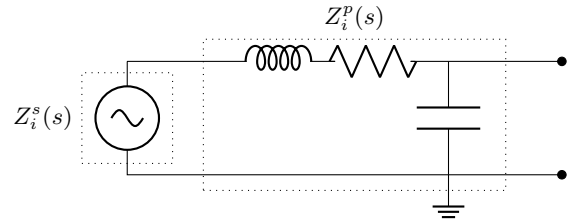
\begin{figure}[t]
\centering
\begin{circuitikz}
\small
    \draw
        (0,1.5) 
        to[short] ++(1.5,0)
        to[cute inductor] ++(1,0)
        to[short] ++(0.25,0) 
        to[american resistor] ++(1,0) 
        to[short,-*] ++(2,0)
        to[open,-*] ++(0,-1.5)
        to[short] (0,0)
        to[sV,name=vs] (0,1.5)
        {};         
    \draw
        (4.25,1.5) to[capacitor] ++(0,-1.5) node[ground] 
        {};

    \node[draw, rectangle, dotted, anchor=center,
        label=left:$Z_i^s(s)$, 
        minimum width=1.2cm, 	minimum height=1.2cm]
        (Zs) at (vs.center) {};

    \node[draw, rectangle, dotted, anchor=north west,
        label=above:$Z_i^p(s)$, 
        minimum width=4cm, 	minimum height=2.0cm]
        (Zp) at (1,1.8) {};
   
\end{circuitikz}
\caption{A series connection of a voltage source $Z_i^s(s)$ with a passive impedance $Z_i^p(s)$.}
\label{fig:tikz_series_impedance}
\end{figure}

\begin{lemma} \label{lemma:pq_unstable}
Consider a passive impedance at $\nu_i \in \mathcal{V}_B$ with $IV$ transfer function $Z_i^p(s)$ as in \eqref{eq:gen_passive_compent}. The eigenvalues of the return ratio $L_{G,i}(s) = \tilde{Z}_i^p(s)W_i$, with $\tilde{Z}_i^p(s) = \frac{1}{{V_i^\star}^4}U_i^\ddagger Z_i^p(s)U_i^\dagger$, $U_i^\dagger$ as in \eqref{eq:pq_deriv_1}, $U_i^\ddagger$ as in \eqref{eq:pq_deriv_2}, and $W_i$ as in \eqref{eq:pq_deriv_3}, are given by
\begin{equation} \label{eq:instab_pq_loci}
    \lambda_\pm(s) = \pm \frac{S_i^\star}{ {V_i^\star}^2 } \sqrt{a(s)^2 + b(s)^2},
\end{equation}
where $S_i^\star := \sqrt{ {P_i^\star}^2 + {Q_i^\star}^2 }$.
\end{lemma}
\begin{proof}
By direct calculation, we get
\begin{equation*}
    \tilde{Z}_i^p(s) = \frac{1}{ {V_i^\star}^2 } \begin{bmatrix}
        -b(s) & -a(s) \\ a(s) & -b(s)
    \end{bmatrix}.
\end{equation*}
Therefore
\begin{equation*}
    L_{G,i}(s) = \frac{1}{ {V_i^\star}^2 } \begin{bmatrix}
        - a(s)P_i^\star +b(s)Q_i^\star  & -a(s)Q_i^\star-b(s)P_i^\star  \\ -a(s)Q_i^\star - b(s)P_i^\star & a(s)P_i^\star -b(s)Q_i^\star
    \end{bmatrix}.
\end{equation*}
The eigenvalues of $L_{G,i}(s)$ are given by the solutions to $\det(\lambda I-L_{G,i}(s)) = 0$. This gives
\begin{equation*}
    \begin{aligned}
        &\lambda^2 \!-\! \frac{1}{ {V_i^\star}^4 } \left[(a(s)P_i^\star \!-\! b(s)Q_i^\star)^2 \!+\! (a(s)Q_i^\star \!+\! b(s)P_i^\star)^2 \right] = 0 \\
        &\implies \lambda^2 - \frac{{S_i^\star}^2}{ {V_i^\star}^4 } \left[ a(s)^2 + b(s)^2 \right] = 0.
    \end{aligned}
\end{equation*}
Solving the above gives \eqref{eq:instab_pq_loci}.
\end{proof}

\begin{figure}[t]
    \centering
    \begin{subfigure}{0.45\textwidth}
        \centering
        \includegraphics[width=0.95\textwidth]{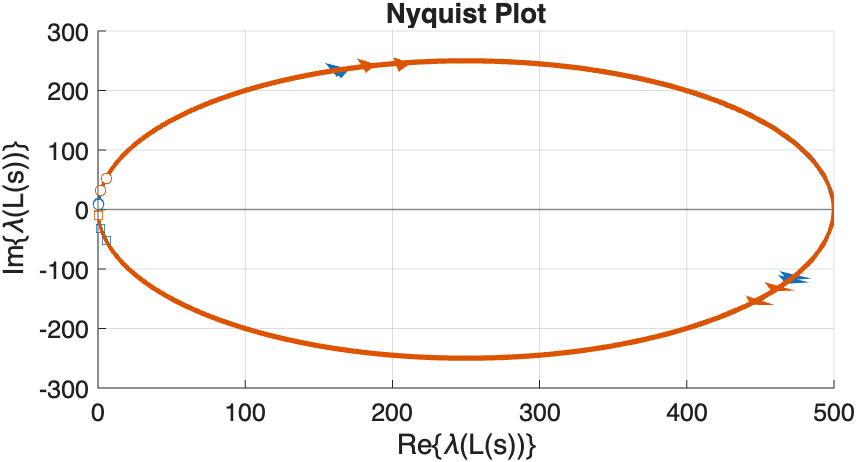}
        \caption{Nyquist plot of $Z_i^p(s)$.}
        \label{fig:plot_Zpassive}
        \vspace*{1em}
    \end{subfigure}
    \\
    \begin{subfigure}{0.45\textwidth}
        \centering
        \includegraphics[width=0.95\textwidth]{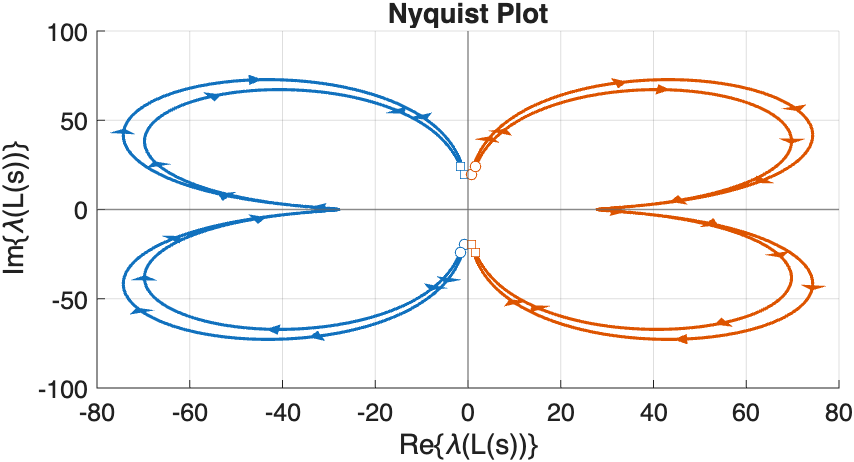}
        \caption{Nyquist plot of $L_{G,i}(s) = -\tilde{Z}_i^p(s)W_i$}
        \label{fig:plot_pq_unstable_sub}
    \end{subfigure}
    \caption{The Nyquist plots of $Z_i^p(s)$ and $L_{G,i}(s) = -\tilde{Z}_i^p(s)W_i$ for the $RLC$ circuit shown in \cref{fig:tikz_series_impedance} over the frequency band \SIrange{4e3}{5e3}{\radian\per\second} (frequencies near the resonant frequency).
    }
    \label{fig:plots_pq_unstable}
\end{figure}

We therefore see that the characteristic loci of $L_{G,i}(s)$ are reflected through the origin. Typically, the Nyquist plot of $Z_i^p(s)$ for a strictly passive circuit forms a large arc in \lhb{$\mathbb{C}_+$} near its resonant frequency. For example, \cref{fig:plot_Zpassive} shows the Nyquist plot of the circuit shown in \cref{fig:tikz_series_impedance} with $R = 0.025$ p.u., $X = 0.25$ p.u., and $B = 0.02$ p.u.~(with nominal grid frequency $\omega_0 = 2\pi \times \SI{50}{\radian\per\second}$) in the frequency range \SIrange{4e3}{5e3}{\radian\per\second}. \Cref{fig:plot_pq_unstable_sub} then shows the corresponding Nyquist plot of $L_{G,i}(s)=-\tilde{Z}_i^p(s)W_i$ (with negative sign to account for the positive-feedback interconnection in \eqref{eq:pq_from_iv}) in the same frequency band using the power and voltage setpoints from the example in \cref{sec:mod_iv_instabilities}. This clearly shows a set of branches symmetric about the origin and \lhb{two clockwise} encirclements of the point $-1$, indicating that
\lhb{the $PQ$ transfer function of the $RLC$ element has two poles in $\mathbb{C}_+$.} By direct calculation, we find that $G_i(s)$ has right half-plane poles at $s = 1.1 \times 10^3$ and $s = 1.6 \times 10^4$.

\begin{remark}
    \lhb{As in \cref{sec:mod_iv_instabilities}, the full bus dynamics including $Z_i^s(s)$ should be considered to conclude the stability properties of $G_i(s)$. However, if the approximation \eqref{eq:Z_bus_approx_Z_passive} holds over a frequency band where $Z_i^p(s)$ undergoes a large phase change (producing a large arc in the Nyquist plot), then \cref{lemma:pq_unstable} indicates that $L_{G,i}(s) = -\frac{1}{{V_i^\star}^4} U_i^\ddagger(Z_i^s(s)+Z_i^p(s))U_i^\dagger W_i$ may trace a corresponding arc in the left half-plane, contributing an encirclement of the point $-1$. }
\end{remark}
\begin{remark}
    \lhb{It should be noted that the configuration shown in \cref{fig:tikz_series_impedance} represents many common scenarios, such as a synchronous generator in series with a transformer, or a grid-forming or grid-following inverter with an $RLC$-based filter. The dynamics of these high-frequency components are often ignored in power system stability studies involving $PQ$ models. However, \cref{lemma:pq_unstable} shows that these components should be \icl{treated} 
    with caution, as they may contribute unstable poles which must be stabilised by the interconnection.
    }
\end{remark}

\subsection{Example Involving More Complex Dynamics}
We now show using a more involved example that the phenomena identified using simplified dynamics in \cref{lemma:iv_unstable,lemma:pq_unstable} hold for more complicated systems.

\begin{figure}[t]
    \centering
    \includegraphics[width=0.48\textwidth]{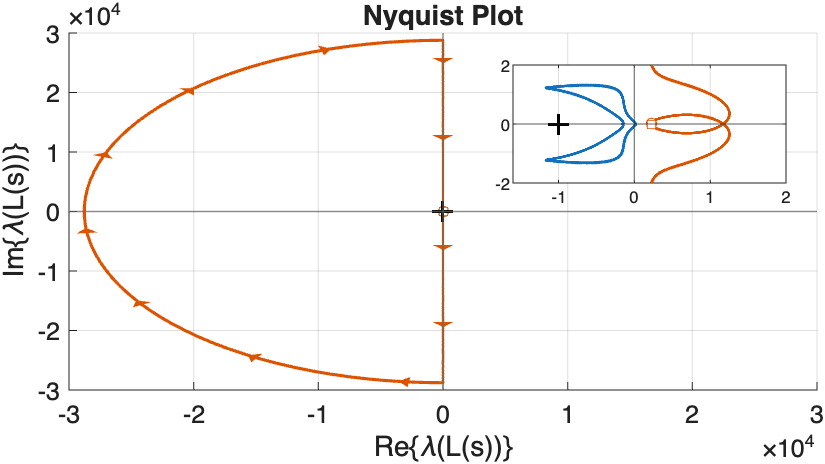}
    \caption{The Nyquist plot of $L_{Z,i}(s) = G_i(s)W_i$ along the modified Nyquist contour (with indentation of radius $\epsilon = 1\times 10^{-4}$ at the origin), where $G_i(s)$ is modelled as a synchronous generator with parameters matching generator 1 in Example 12.6 of \cite{kundur_PowerSystem_22}. The inset rescales the plot to show the bounded branch.}
    \label{fig:plot_kundur_iv}
\end{figure}

Consider a bus $\nu_i \in \mathcal{V}_B$ containing a synchronous generator.
The synchronous generator \lhb{model follows} Example 12.6 of \cite{kundur_PowerSystem_22}, \icl{a widely used example in power system stability analysis. The  parameters} and operating point \icl{match} \lhb{those} at bus 1 in the \lhb{multi-machine} example, \icl{and we also note that the aggregate interconnection of the multi-machine system is stable.}
 The generator is equipped with a thyristor exciter with high transient gain and power system \icl{stabiliser (PSS), 
and the} frequency is controlled using a 4\% droop constant.
The $PQ$ model of the system is linearised about its operating point to obtain a stable transfer function $G_i(s)$.

We can then compute the $IV$ impedance for the generator $Z_i^s(s)$ by constructing the matrices $U_i^\dagger$, $U_i^\ddagger$ and $W_i$ using the generator operating point and using \eqref{eq:iv_from_pq}. By direct calculation, we find that $Z_i^s(s)$ has a right half-plane pole at $s = 5.9$. The fact that there is an unstable pole can be verified by applying \lhb{\cref{thm:nyquist_std}} to $L_{Z,i}(s) = G_i(s) W_i$, as shown in \cref{fig:plot_kundur_iv}. As in the simple case considered in \cref{sec:mod_iv_instabilities}, we see that the unbounded branch of the characteristic loci encircles the point $-1$.


Next, we place the generator impedance $Z_i^s(s)$ obtained above in series with a step-up transformer with impedance $Z_i^p(s)$, as in \cref{fig:tikz_series_impedance}. The transformer is assumed ideal, with parameters matching that linking bus 1 to 5 in Example 12.6 of \cite{kundur_PowerSystem_22}. In addition, we include the shunt capacitance of the Pi-model of the line linking bus 5 to 6 within the impedance $Z_i^p(s)$. The impedance of the aggregate bus can then be calculated as $Z_i^{\mathrm{tot}}(s) = Z_i^s(s)+Z_i^p(s)$, which inherits one right half-plane pole from $Z_i^s(s)$.

\lhb{Using the voltage and current at the transformer interface with the rest of the grid, we compute the operating point quantities needed to construct}
the matrices $U_i^\dagger$, $U_i^\ddagger$ and $W_i$ for the new aggregate bus.
Using these matrices, the aggregate bus $PQ$ transfer function $G_i^{\mathrm{tot}}(s)$ can be computed using \eqref{eq:pq_from_iv}. By direct calculation, we find that $G_i^{\mathrm{tot}}(s)$ has right half-plane poles at $s = 3.1 \times 10^3$ and $s = 1.2 \times 10^5$.

The number of unstable poles in $G_i^{\mathrm{tot}}(s)$ can be verified by applying the multivariable Nyquist criterion to $L_{G,i}(s)= - \tilde{Z}_i^{\mathrm{tot}}(s)W_i$ (with a minus sign to account for the positive-feedback interconnection in \eqref{eq:pq_from_iv}), 
which is shown in \cref{fig:plot_kundur_pq}. Here, we see one counter-clockwise encirclement and two clockwise encirclements, giving a net encirclement of one. As $\tilde{Z}_i^{\mathrm{tot}}(s)$ has one unstable pole, we can therefore conclude that $G_i^{\mathrm{tot}}(s)$ has two unstable poles \lhb{using \eqref{eq:nyquist_pole_count}}.

\begin{figure}[t]
    \centering
    \begin{subfigure}{0.45\textwidth}
        \centering
        \includegraphics[width=0.95\textwidth]{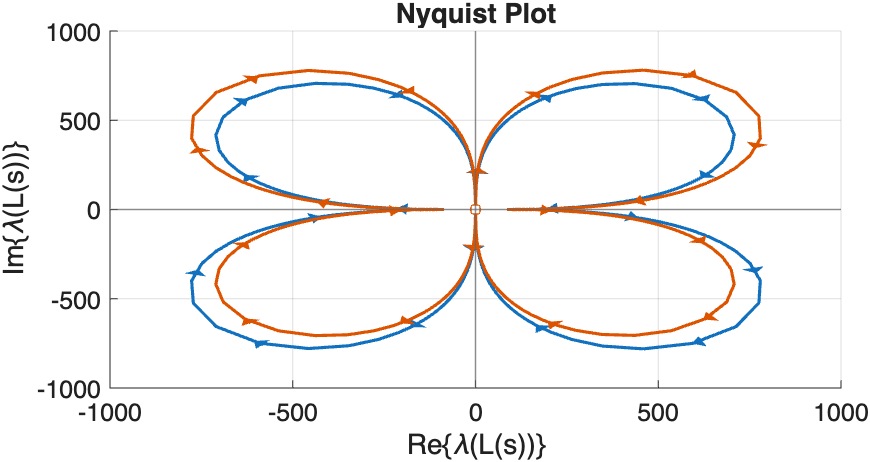}
        \caption{Full Nyquist plot of $L_{G,i}(s)$.}
        \label{fig:plot_kundur_pq1}
        \vspace*{1em}
    \end{subfigure}
    \\
    \begin{subfigure}{0.45\textwidth}
        \centering
        \includegraphics[width=0.95\textwidth]{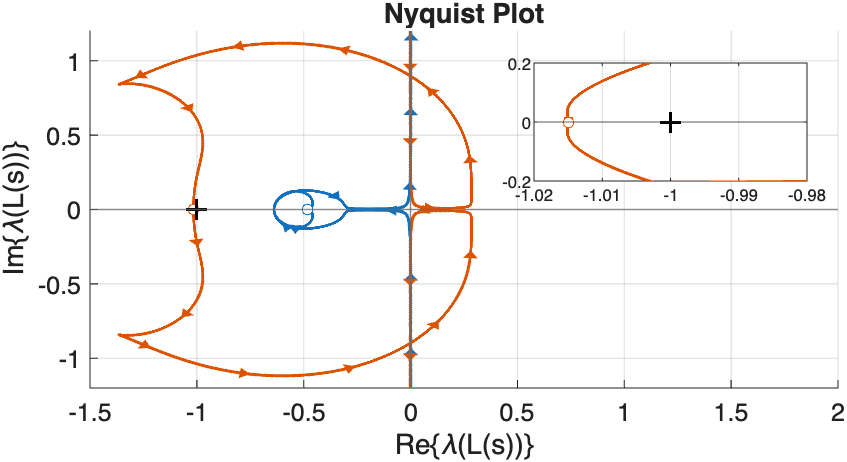}
        \caption{Zoomed in Nyquist plot of $L_{G,i}(s)$.}
        \label{fig:plot_kundur_pq2}
    \end{subfigure}
    \caption{The Nyquist plot $L_{G,i}(s) = -\tilde{Z}_i^{\mathrm{tot}}(s)W_i$ for a synchronous generator in series with a transformer. The first plot shows the full scale of the Nyquist plot, while the second shows a zoomed in picture near the origin, with an inset to show the plot near the point $-1$.
    }
    \label{fig:plot_kundur_pq}
\end{figure}
\section{Conclusion} \label{sec:conclusion}
In this paper, we derived two equivalent small-signal representations of \icl{AC grids, as interconnections of subsystems,} and showed how a loop transformation can be used to translate from one \icl{representation} 
to the other. We then showed analytically and numerically that unstable dynamics may appear in either representation in common scenarios. Our results highlight that
\lhb{simplifying the dynamics of bus components should be done with caution,}
as high frequency dynamics or device setpoints may introduce instabilities that must be accounted for in stability analysis and system identification. \icl{This issue is particularly important when a decentralised stability analysis is carried out for grid-code formulations, as unstable subsystems are much more difficult to be accounted for, relying often on interconnection with other subsystems for their stabilisation.}

\bibliographystyle{IEEEtran}
\bibliography{./Files/references.bib}
\end{document}